\documentclass{article}
\usepackage{algorithm}
\usepackage{algpseudocode}
\usepackage{amsmath}
\usepackage[square, numbers, sort&compress]{natbib}
\usepackage{caption}
\usepackage{amsthm}
\usepackage{multirow}

\newtheorem{theorem}{Theorem}
\newtheorem{definition}{Definition}
\newtheorem{lemma}{Lemma}
\newtheorem{corollary}{Corollary}
\newtheorem{assumption}{Assumption}
\newtheorem{proposition}{Proposition}
\newtheorem{remark}{Remark}

\usepackage{subcaption}
\usepackage{graphicx}

\usepackage{amssymb}
\usepackage[mathscr]{euscript}

\usepackage{natbib}

\usepackage{amsfonts,amsmath,amssymb}
\usepackage{bbm}
\usepackage{xspace}

\usepackage{multicol}
\usepackage{enumitem}

\newcommand{\argmin}{\operatornamewithlimits{argmin}}

 \usepackage[main,final]{neurips/neurips_2026}

\usepackage[utf8]{inputenc} % allow utf-8 input
\usepackage[T1]{fontenc}    % use 8-bit T1 fonts
\usepackage{hyperref}       % hyperlinks
\usepackage{url}            % simple URL typesetting
\usepackage{booktabs}       % professional-quality tables
\usepackage{amsfonts}       % blackboard math symbols
\usepackage{nicefrac}       % compact symbols for 1/2, etc.
\usepackage{microtype}      % microtypography
\usepackage{xcolor}         % colors

\title{Control Charts for Multi-Agent Systems}

\author{%
Hayden Helm\thanks{Corresponding author: \texttt{hayden@helivan.io}.} \\ Helivan
\And
Carey Priebe \\
Johns Hopkins University
\And
Brandon Duderstadt \\
Calcifer Computing
}

\begin{document}

\maketitle

\begin{abstract}
Generative agents have proven to be powerful assistants in a wide variety of contexts.
Given this success, users are now deploying agents with minimal restrictions in open ended, multi-agent environments.
Current methods for monitoring the dynamics of open-ended multi-agent systems are limited to aggregate statistics, qualitative inspection, or detecting change between two fixed timepoints.
In this paper, we extend the process-theoretic notion of adaptive control charts to multi-agent systems to enable automated monitoring.
Using simulation, we demonstrate that adaptive control charts are necessary for monitoring multi-agent systems that can learn from their environment.
We further demonstrate, both empirically and theoretically, that adaptive control charts are susceptible to adversarial agents that defect sufficiently slowly.
These results suggest a fundamental tradeoff in multi-agent system control: monitors with fixed baselines flag legitimate learning, while monitors with adaptive baselines admit slow adversaries.
We conjecture this holds for any monitor.
\end{abstract}

\section{Introduction}
A generative agent is a system whose behavior can be affected by a change in its context, its tools, its base model, or its environment \citep{park2023generative, bridgeford2025detectingperspectiveshiftsmultiagent}.
For example, consider OpenClaw \citep{steinberger2025openclaw}, an open-source personal AI assistant that runs locally, maintains persistent memory, and can execute tasks autonomously via various interfaces. 
OpenClaw -- or any sufficiently capable agent \citep{kinniment2023evaluating} -- can navigate the internet, update its own knowledge base, and take actions on behalf of its user without explicit step-by-step instruction.

Generative agents are increasingly deployed in environments populated by other agents performing similar tasks.
Moltbook \citep{schlichtm2025moltbook}, a social network designed explicitly for OpenClaw instances to interact, exemplifies this trend.
On Moltbook, agents sign up, post, discuss, and upvote content with no direct human involvement.
Such environments are naturally modeled as multi-agent systems in which agents interact, exchange information, and influence each other's behaviors.
Moltbook is a caricature of more complex environments with similar properties, such as the internet, where agents and humans increasingly coexist and influence one another, but where controlled study is difficult.

Given the proliferation of generative agents in shared environments, developing methods for monitoring the dynamics of the system is critical to the reliability and safety of open environments.
Current approaches for measuring system behavior are limited to aggregate agent-level statistics or qualitative inspection \citep{park2023generative, chuang2024simulating}.
Recent work on change detection in multi-agent systems \citep{bridgeford2025detectingperspectiveshiftsmultiagent} addresses only the zeroth-order problem of determining whether an a group of agent's behavior changed between two timepoints.
These methods are not easily extended to dynamic environments, as any natural drift will be flagged as a change.

We make three distinct contributions:
\begin{enumerate}
    \item     \item \textbf{Adaptive control for black-box multi-agent systems.} We introduce adaptive control charts for system-level multi-agent monitoring based on the Temporal Data Kernel Perspective Space \citep{bridgeford2025detectingperspectiveshiftsmultiagent} and the iso-mirror \citep{athreya2025euclidean}. To our knowledge, these are the first control charts for multi-agent systems that adapt to natural environmental drift;
    \item \textbf{Necessity of adaptive control in dynamic systems / environments.} We demonstrate empirically that adaptive control is necessary for system monitoring in dynamic environments;
    \item \textbf{Vulnerability of adaptive control to adversarial agents.}     \item \textbf{Vulnerability of adaptive control to adversarial agents.} We demonstrate empirically and theoretically that adaptive control of multi-agent systems is vulnerable to adversarial agents that defect sufficiently slowly.
\end{enumerate}

\section{Related Work}

\paragraph{Monitoring multi-agent systems.}
The deployment of generative agents is accelerating \citep{yee2024agents}, raising concerns about emergent risks in multi-agent settings \citep{hammond2025multiagent, tomasev2025distributional, shapira2026agentschaos}.
\citet{tomasev2025distributional} argue that general intelligence may emerge from interacting sub-AGI agents rather than a monolithic system, while \citet{hammond2025multiagent} catalog failure modes including collusion, coordination failures, and cascading effects.
Empirical studies have examined agent dynamics in sandbox environments \citep{park2023generative}, opinion formation \citep{chuang2024simulating}, and social alignment via mirroring \citep{mcguinness2025mirroring}.
\citet{helm2024tracking} introduced the perspective space for tracking multi-agent system evolution over time and showed it was correlated to causal interventions in the agent-to-agent communication network.
Previous works report aggregate statistics or qualitative observations but do not provide methods for automated system-level monitoring.

\paragraph{Low-dimensional representations of models and agents.}
Comparing generative systems requires tractable representations that capture behaviorally relevant differences.
Prior work has embedded models via internal activations \citep{duderstadt2024comparing, huh2024platonicrepresentationhypothesis}, parameter weights \citep{putterman2024learning}, or responses to reference queries \citep{acharyya2025consistent}.
The Data Kernel Perspective Space (DKPS) embeds models based on response similarities and supports statistical inference in the black-box setting \citep{acharyya2025consistent, helm2025inference}.
\citet{bridgeford2025detectingperspectiveshiftsmultiagent} extended DKPS to temporal settings (TDKPS) and developed hypothesis tests for detecting behavioral change between two timepoints.
We use TDKPS as the basis for automated system-level monitoring.

\paragraph{Change detection and control charts.}
Classical control charts monitor statistics against fixed thresholds to detect process deviations \citep{shewhart1931economic, page1954continuous}.
Modern extensions handle multivariate and high-dimensional data \citep{qiu2013statistical} or use deep learning for time series anomaly detection \citep{schmidl2022anomaly, choi2024timeseries}.
For network data, spectral embedding methods detect anomalous graphs and vertices in time series \citep{chen2020mase}, while latent position models localize changepoints via Euclidean mirrors \citep{chen2024euclidean}.
These methods assume a stationary baseline or require labeled anomalies for training.
In multi-agent systems, agents naturally drift as they interact, rendering fixed baselines useless.
To our knowledge, no prior work develops control charts for multi-agent systems that adapt to natural environmental drift.
\section{Methods}

We develop control charts for monitoring multi-agent systems by applying the iso-mirror \citep{athreya2025euclidean} to the Temporal Data Kernel Perspective Space (TDKPS) \citep{bridgeford2025detectingperspectiveshiftsmultiagent} representations of agents.
We first describe our setting and then detail the methodology.
We assume that the system-level monitoring system has (black-box) access to individual agents \citep{bridgeford2025detectingperspectiveshiftsmultiagent}.

\paragraph{Setting.}
We observe $N$ agents $\mathcal{F}^{(t)} = \{f_1^{(t)}, \ldots, f_N^{(t)}\}$ over $T$ timepoints.
In the black-box setting, each agent $f_n^{(t)}: \mathcal{Q} \to \mathcal{X}$ is a random mapping from a query space $\mathcal{Q}$ to a response space $\mathcal{X}$.
Given query $q$, the response $f_n^{(t)}(q)$ is sampled from distribution $P_n^{(t)}(q)$.
Let $g: \mathcal{X} \to \mathbb{R}^p$ be an embedding function mapping responses to real vectors\footnote{TDKPS representations depend on the choice of $g$. \citet{bridgeford2025detectingperspectiveshiftsmultiagent} found that they highly correlated across modern embedding functions.}.
For queries $Q = \{q_1, \ldots, q_M\}$ with $r$ responses per query, we denote the average embedded response matrix as $\bar{X}_n^{(t)} \in \mathbb{R}^{M \times p}$, where the $m$-th row is $\bar{X}_{nm\cdot}^{(t)} := \frac{1}{r} \sum_{k=1}^r g(f_n^{(t)}(q_m)_k)$.

\subsection{Temporal Data Kernel Perspective Space}

TDKPS \citep{bridgeford2025detectingperspectiveshiftsmultiagent} embeds all agent-timepoint pairs into a shared low-dimensional Euclidean space based on response similarities.
We construct the temporal distance matrix $D \in \mathbb{R}^{NT \times NT}$ with entries
% \begin{equation*} 
$ D_{(t,n),(t',n')} = \frac{1}{M} \|\bar{X}_n^{(t)} - \bar{X}_{n'}^{(t')}\|_F $,
% \end{equation*}
capturing behavioral differences between agent $n$ at time $t$ and agent $n'$ at time $t'$.

The $d$-dimensional TDKPS embeddings $\{\psi_n^{(t)} : n \in [N], t \in [T]\}$ are defined as a solution to
\begin{equation}\label{eq:tdkps}
\{\psi_n^{(t)}\} = \argmin_{z_{n}^{(t)} \in \mathbb{R}^d} \sum_{(t,n),(t',n')} \left( \|z_n^{(t)} - z_{n'}^{(t')}\| - D_{(t,n),(t',n')} \right)^2.
\end{equation}
This joint embedding captures which agents are similar at each timepoint and how each agent evolves over time.
When $D$ is a Euclidean distance matrix, Eq.~\eqref{eq:tdkps} can be solved efficiently in closed form via classical multidimensional scaling \citep{torgerson1952multidimensional}.

\subsection{The iso-mirror}

While TDKPS provides agent-level embeddings, monitoring large multi-agent systems benefits from a system-level summary.
The iso-mirror \citep{athreya2025euclidean} provides such a summary by embedding the system itself based on how the system state differs across time.

Given TDKPS embeddings $\{\psi_n^{(t)} : n \in [N]\}$ at each timepoint $t$, we construct a timepoint-by-timepoint distance matrix $\Delta \in \mathbb{R}^{T \times T}$ with entries
\begin{equation*}
\Delta_{t,t'} = \delta\left(\{\psi_n^{(t)}\}_{n=1}^N, \{\psi_n^{(t')}\}_{n=1}^N\right),
\end{equation*}
where $\delta(\cdot, \cdot)$ measures dissimilarity between the agent embedding distributions at times $t$ and $t'$.

The $c$-dimensional iso-mirror $\{\phi^{(t)} : t \in [T]\}$ is obtained by applying MDS to $\Delta$:
\begin{equation*}
\{\phi^{(t)}\} = \argmin_{z^{(t)} \in \mathbb{R}^c} \sum_{t, t'} \left( \|z^{(t)} - z^{(t')}\| - \Delta_{t,t'} \right)^2.
\end{equation*}
The iso-mirror trajectory $\{\phi^{(t)}\}_{t=1}^T$ provides a low-dimensional representation of system-level dynamics, where proximity in iso-mirror space indicates similar system states.
Depending on the choice of dissimilarity, the iso-mirror can be sensitive to changes in dispersion, shape, or higher-order structure of the agent distribution and not merely shifts in location.
As defined, the TDKPS and iso-mirror are fit jointly on all $T$ timepoints, so $\phi^{(t)}$ depends on observations after $t$. We discuss keeping the embeddings fixed for online monitoring in \S~\ref{sec:discussion}.

\subsection{Multi-agent system control charts via the iso-mirror}
\label{sec:control-charts}
We apply the Shewhart control chart framework \citep{shewhart1931economic} to the iso-mirror trajectory to automate multi-agent system monitoring.
Shewhart control charts track a process by comparing observations to control limits derived from historical data. 
Observations falling outside these limits signal potential anomalies.

Let $S: \mathbb{R}^c \to \mathbb{R}$ be a univariate statistic defined on the iso-mirror.
Common choices include the the norm $S(\phi^{(t)}) = \|\phi^{(t)}\|_2$ and the norm of the first difference $S(\phi^{(t)}) = \|\phi^{(t)} - \phi^{(t-1)}\|_{2}$.

\paragraph{Establishing control limits.}
Control limits are estimated from a reference window $\mathcal{W}^{(t)} \subseteq \{1, \ldots, t-1\}$ of past observations presumed to be in-control.
From the reference observations $\{S(\phi^{(\tau)}) : \tau \in \mathcal{W}^{(t)}\}$, we estimate
\begin{equation*}
\hat{\mu}^{(t)} = \frac{1}{|\mathcal{W}^{(t)}|} \sum_{\tau \in \mathcal{W}^{(t)}} S(\phi^{(\tau)}), \qquad \hat{\sigma}^{(t)} = \sqrt{\frac{1}{|\mathcal{W}^{(t)}| - 1} \sum_{\tau \in \mathcal{W}^{(t)}} \|S(\phi^{(\tau)}) - \hat{\mu}^{(t)}\|_2^2}.
\end{equation*}
Control limits at time $t$ are
\begin{equation*}
\text{LCL}^{(t)} = \hat{\mu}^{(t)} - k \cdot \hat{\sigma}^{(t)} \qquad \text{and} \qquad \text{UCL}^{(t)} = \hat{\mu}^{(t)} + k \cdot \hat{\sigma}^{(t)},
\end{equation*}
and we signal an alarm if $S(\phi^{(t)}) \notin [\text{LCL}^{(t)}, \text{UCL}^{(t)}]$.
$k > 0$ controls the trade-off between sensitivity and false alarm rate.
Two choices for $\mathcal{W}^{(t)}$ yield qualitatively different monitoring regimes:
\begin{itemize}
    \item \emph{Fixed control}: $\mathcal{W}^{(t)} = \{1, \ldots, T_0\}$ for all $t > T_0$. Control limits are fixed based on an initial burn-in period and remain constant throughout monitoring.
    \item \emph{Adaptive control}: $\mathcal{W}^{(t)} = \{t - w, \ldots, t - 1\}$ for window size $w$. Control limits update at each timepoint, adapting to non-stationary baseline dynamics.
\end{itemize}
Fixed control is appropriate when the in-control distribution is stable. 
Adaptive control accommodates systems with natural drift.
Some statistical properties of control, including Type-I and Type-II error rates, are discussed in \S~\ref{sec:theory} and Appendix \ref{app:proofs}.
\section{Experiments}
\label{sec:experiments}

We simulate a multi-agent system to evaluate the proposed control chart methodology under controlled conditions where ground truth is known. 
We compute the TDKPS in $d = 1$ and the iso-mirror in $c = 1$ via isomap \citep{tenenbaum2000global} applied to the Euclidean distance matrix defined on the perspectives, and monitor $S(\phi^{(t)}) = \phi^{(t)}$.

\paragraph{Queries and environment.}
Agents interact over a set of $M$ queries $Q = \{q_1, \ldots, q_M\}$, partitioned into two types based on their ground-truth dynamics.
Static queries ($M_s$)  are queries whose correct responses remain fixed throughout the simulation.
Environmental queries ($M_{\text{env}}$) are queries whose correct responses are functions of the environment and evolve over time, updating with probability $p_{\text{env}}$ per iteration.
The environmental queries reflect real-world settings where ground truth changes as the environment evolves (e.g. ``What is the weather today?", ``Is NVDA stock up?", etc.). 
This distinction motivates an agent architecture capable of updating its knowledge base even when an entry for a query already exists, as agents that treat their initial answers as immutable will become increasingly misaligned with the environment over time.
Unless otherwise stated, we set $M = 50$, $M_s = 40$, $M_{\text{env}} = 10$, and $p_{\text{env}} = 0.04$.

\paragraph{Agent architecture.}
Each agents' cognitive architecture consists of a large language model equipped with external vector database of memories.
At initialization, each agent's vector database is populated with $M_s/N$ static question-answer pairs and $M_{env}/N$ dynamic question-answer pairs. 
Every agent in the system is initialized with a disjoint set of question-answer pairs.
Database entries for dynamic question-answer pairs update as their ground truth evolves for the agents that were assigned them at initialization.
The only way that agents in the swarm can learn the current answers to dynamic questions they were not initialized with is to communicate with other agents.

At inference time, agents retrieve the $k=3$ most relevant question-answer pairs from their database before generating a response (see Appendix \ref{app:prompt-template} for prompt template and examples). 

Agents are instructed to be "grounded," in that they are instructed to respond "I don't know" if they do not have sufficient context to support an answer to a given query.
We use OpenAI's \texttt{gpt-4o-mini} \citep{openai2024gpt4o} as the base model unless otherwise stated.

\paragraph{Retrieval mechanism}
We populate each agent's vector database with question-answer pairs by vectorizing questions with \texttt{nomic-embed-text-v1.5} \cite{nussbaum2025nomicembedtrainingreproducible}.
Inspired by 
% the cognitive architecture in
\citet{park2023generative}, we consider both memory content and recency during retrieval.
In particular, we measure the relevance between a vectorized query $q$ and a vector database entry $k$ as $ \langle q, k \rangle + e^{(- \rho \Delta t_{k}) }$, where $\Delta t_{k}$ measures the time since the memory was inserted into the database, and $\rho$ controls the relative impact of query recency.
We refer to an agent with $\rho>0$ \textit{dynamic} and an agent with $\rho=0$ \textit{static}.

\paragraph{Interaction and update mechanisms.}
Agents are connected via a complete communication graph with uniform edge weights $w_{nn'} = 1/(N-1)$ for $n \neq n'$.
At each iteration $t$, agent $f_n^{(t)}$ selects another agent $f_{n'}^{(t)}$ to interact with by sampling from the edge weight distribution.

During an interaction, agent $f_n^{(t)}$ poses $K$ questions to agent $f_{n'}^{(t)}$, sampled from $Q$ with probabilities proportional to the time elapsed since each question was last asked,
$P(q_{nm}) = \exp(\ell_{nm})/\sum_{m'=1}^{M} \exp(\ell_{nm'})$,
where $\ell_m$ denotes the number of iterations since agent $n$ last asked question $m$.
This recency-weighted sampling encourages agents to seek updated information on questions they have not recently queried.
In the dynamic case, agent $f_n^{(t)}$ stores any substantive responses in its knowledge base by storing the question it asked paired with the answer its received.
In the static case, agent $f_n^{(t)}$ stores a substantive response only if it believes it does not already know the answer to the question.
Unless otherwise stated, we set $N = 5$ and $K = 5$.

\paragraph{Sleeper agents.}
To evaluate the sensitivity of our monitoring procedure to adversarial behavior, we introduce \emph{sleeper agents} into the system \citep{hubinger2024sleeper}.
% A sleeper agent behaves identically to a standard agent until a trigger time $t^*$, after which it begins providing incorrect information to other agents in the network.
Let $\mathcal{F}_{\text{sleeper}} \subseteq \mathcal{F}^{(t)}$ denote the set of sleeper agents and $N_{\text{sleeper}} = |\mathcal{F}_{\text{sleeper}}|$.
For $t < t^*$, sleeper agents participate normally in the communication network.
For $t \geq t^*$, when a sleeper agent
% $f_n^{(t)} \in \mathcal{F}_{\text{sleeper}}$
is queried by another agent, it responds with a canary phrase with probability $P_{defect}(t)$ rather than retrieving from its knowledge base.
We consider a single sleeper agent that defects either quickly or slowly by modulating the $P_{defect}(t)$ schedules.

\begin{figure}
    \centering
    \includegraphics[width=0.95\linewidth]{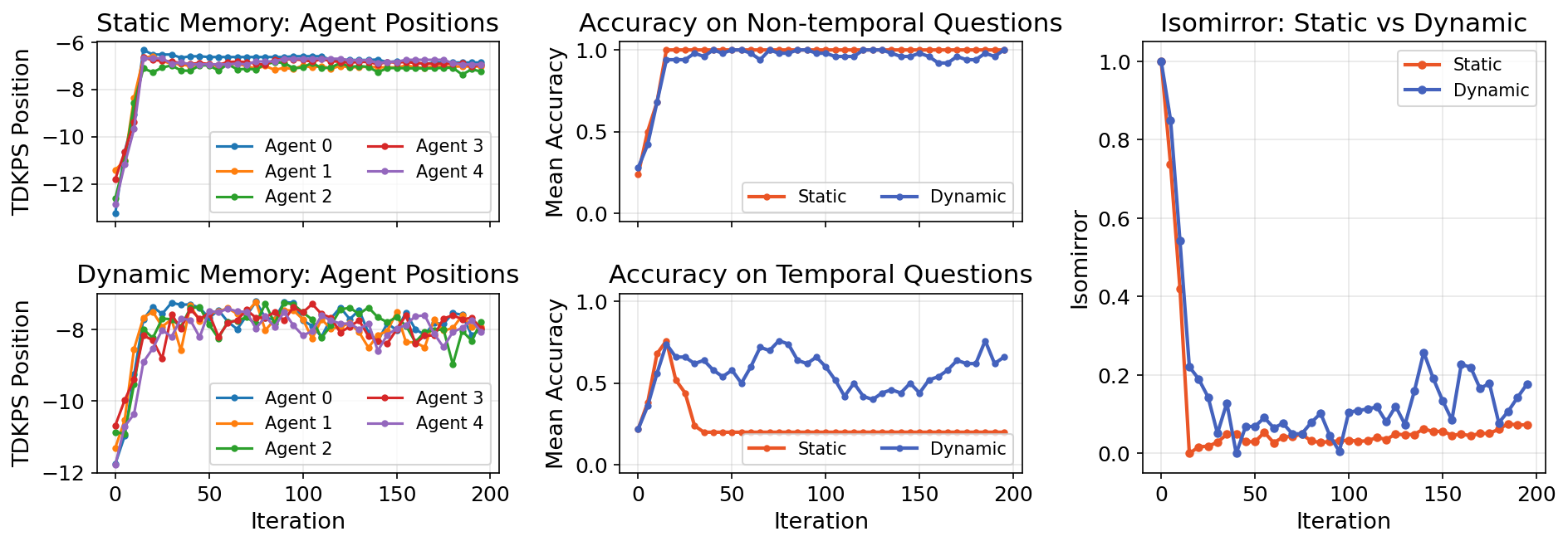}
        \caption{Static ($\rho = 0$) vs.\ dynamic ($\rho = 0.05$) memory. \textbf{Left}: TDKPS embeddings $\psi_n^{(t)}$ for individual agents. \textbf{Middle}: Mean accuracy on static and environmental queries. \textbf{Right}: Iso-mirror trajectories $\phi^{(t)}$. Dynamic memory is necessary for agents to adapt to a dynamic environment.}
    \label{fig:memory}
\end{figure}

\subsection{Results}

\paragraph{Static vs.\ dynamic memory.}
Figure~\ref{fig:memory} compares agents with static memory ($\rho = 0$) against agents with dynamic memory ($\rho = 0.05$).
The left panels show TDKPS embeddings $\psi_n^{(t)}$ for individual agents under each memory type.
With $\rho = 0$, agents converge rapidly to a fixed configuration and remain there indefinitely.
With $\rho = 0.05$, agents exhibit ongoing variation as they incorporate new information from interactions.

The middle panel reveals the effect of memory on accuracy.
On static queries, both memory types achieve near-perfect accuracy.
On environmental queries, agents with $\rho = 0$ quickly fall to $\sim$20\% accuracy (chance-level) as ground truth evolves and their knowledge bases become stale.
Agents with $\rho = 0.05$ maintain substantially higher accuracy on environmental queries, though with greater variance as new information must propagate through the network.

The right panel shows the iso-mirror trajectories \{$\phi^{(t)}\}$.
The static-memory system collapses to a fixed point.
The dynamic-memory system maintains nontrivial variation, capturing ongoing adaptation to environmental change.
These results motivate $\rho > 0 $ for subsequent experiments.

\paragraph{Fixed vs.\ adaptive control.}
Figure~\ref{fig:control} compares fixed and adaptive control charts applied to the iso-mirror trajectory of a system with agents with $\rho = 0.05$.
The left panel shows fixed control with $\mathcal{W}^{(t)} = \{100, \ldots, 300\}$ for all $t > 300$.
As the system evolves beyond the calibration window, $\phi^{(t)}$ drifts systematically below $\hat{\mu}^{(t)}$, producing repeated violations of $\text{LCL}^{(t)}$.
These violations reflect natural system evolution rather than anomalous behavior.

The right panel shows adaptive control with $\mathcal{W}^{(t)} = \{t - w, \ldots, t - 1\}$ and $w = 200$.
The control limits track the evolving baseline, and the iso-mirror trajectory remains largely within the control bands.
Adaptive control avoids the false alarms that plague fixed control in non-stationary settings.
Hence, adaptive control is necessary for controlling dynamic systems with reasonable Type-I error.

\paragraph{Sleeper agent detection.}
Figure~\ref{fig:sleeper} evaluates the control chart's ability to detect sleeper agent attacks.
We characterize different sleeper agents by the speed in which they transition from responding based on their database to responding with the canary phrase for $ t \ge t^{*} $ with probability $ p_{\max}$.
In our experiments, the defection schedule is controlled by a sigmoid of varying duration. 
Explicitly: 
$$P_{defect}(t) \;=\; p_{\max}\cdot \frac{\sigma\!\big(\kappa\,(2(t-t^*)/s - 1)\big) - \sigma(-\kappa)} {\sigma(\kappa) - \sigma(-\kappa)}\,\mathbf{1}_{[t^*,\,t^*+s)}(t)\;+\;p_{\max}\,\mathbf{1}_{[t^*+s,\infty)}(t)$$ 
where $\sigma(x) = (1+e^{-x})^{-1}$ is the standard logistic, $s$ is the attack duration over which the defection rate ramps up, $p_{\max} \in [0,1]$ is the asymptotic per-step probability of emitting an adversarial response after the ramp completes, and $\kappa > 0$ controls the steepness of the ramp.
The top-left panel shows defection with $t^* = 400$, duration $ s = 10 $, $ p_{\max} = 0.75 $, and $\kappa = 1$.
The iso-mirror $\phi^{(t)}$ exhibits a sharp shift, producing multiple consecutive violations of both the $k = 2$ and $k = 3$ control limits.
The attack is clearly detectable.
\begin{figure}[t]
    \centering
    \includegraphics[width=0.95\linewidth]{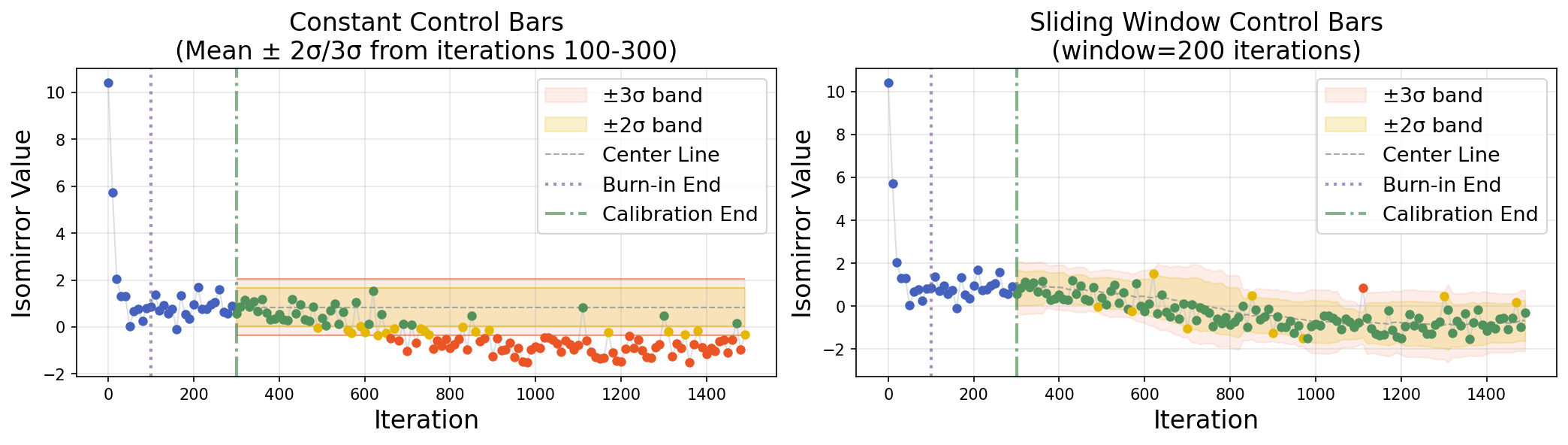}
    \caption{Fixed (\textbf{left}) vs.\ adaptive (\textbf{right}) control charts applied to the iso-mirror $\phi^{(t)}$. Shaded regions indicate $2\sigma$ (yellow) and $3\sigma$ (red) bands. Point colors indicate position relative to control limits: blue (burn-in), green (within $2\sigma$), yellow ($2\sigma$ to $3\sigma$), red (beyond $3\sigma$). Vertical lines mark burn-in end (dotted) and calibration end (dashed). Adaptive control avoids false alarms under natural system drift.}
   \label{fig:control}
\end{figure}

The top-right panel defection with $ t^{*} = 400 $ duration $ s = 800 $, $ p_{\max} = 0.75 $, and $\kappa = 1$. 
The iso-mirror drifts upward, but the adaptive control limits drift with it.
The trajectory remains largely within the control bands, demonstrating that gradual adversarial influence can evade detection.

The bottom panels show the mean accuracy of non-adversarial agents under each attack.
Under immediate defection, accuracy drops sharply but partially recovers.
Under gradual defection, accuracy degrades slowly and persistently, settling at a lower level.
The gradual attack achieves long-term influence while avoiding detection.

\paragraph{Detectability analysis.}
In order to study the probability of detection under various adversarial conditions, we simulate 100 replications of 3 adversarial profiles: no adversary, a fast attack ($s= 100$), and a slow attack ($s= 800$).
Both attacks reach the same peak defection probability $ p_{\max} = 0.75 $ with $\kappa = 1$ and start at $t^* = 400$.

We evaluate adaptive control charts with window size $w = 200$ at sensitivities $k \in \{2, 3, 5\}$.
For each setting, we report the empirical exceedance rate -- the fraction of replications in which the control chart signals an alarm -- at each iteration across the 100 replicates.\footnote{Due to the cost of running several hundred simulations with an LLM in the loop, we replace the LLM call with the proxy function detailed in Appendix~\ref{app:nollm-proxy}.
We show that the LLM-based simulations and proxy-based simulations are qualitatively similar in Figure \ref{fig:validity-of-proxy} of Appendix \ref{app:nollm-proxy}.
Similar proxies have been shown to approximate multi-agent systems faithfully \citep{wang2026gaussianmixturemodelsproxy}.}
Results are shown in Figure~\ref{fig:adversarial-detection-rates}.
Slower defection is results in lower exceedance rates across all sensitivities.  
Theoretical justification for the relative ordering of exceedance rates is provided in \S \ref{sec:theory}.
Appendix~\ref{app:scale} repeats this analysis for systems of up to $100{,}000$ agents on sparse and dense contact graphs, with up to $100$ sleeper agents and a tenfold larger query pool.

\begin{figure}
    \centering
    \includegraphics[width=0.95\linewidth]{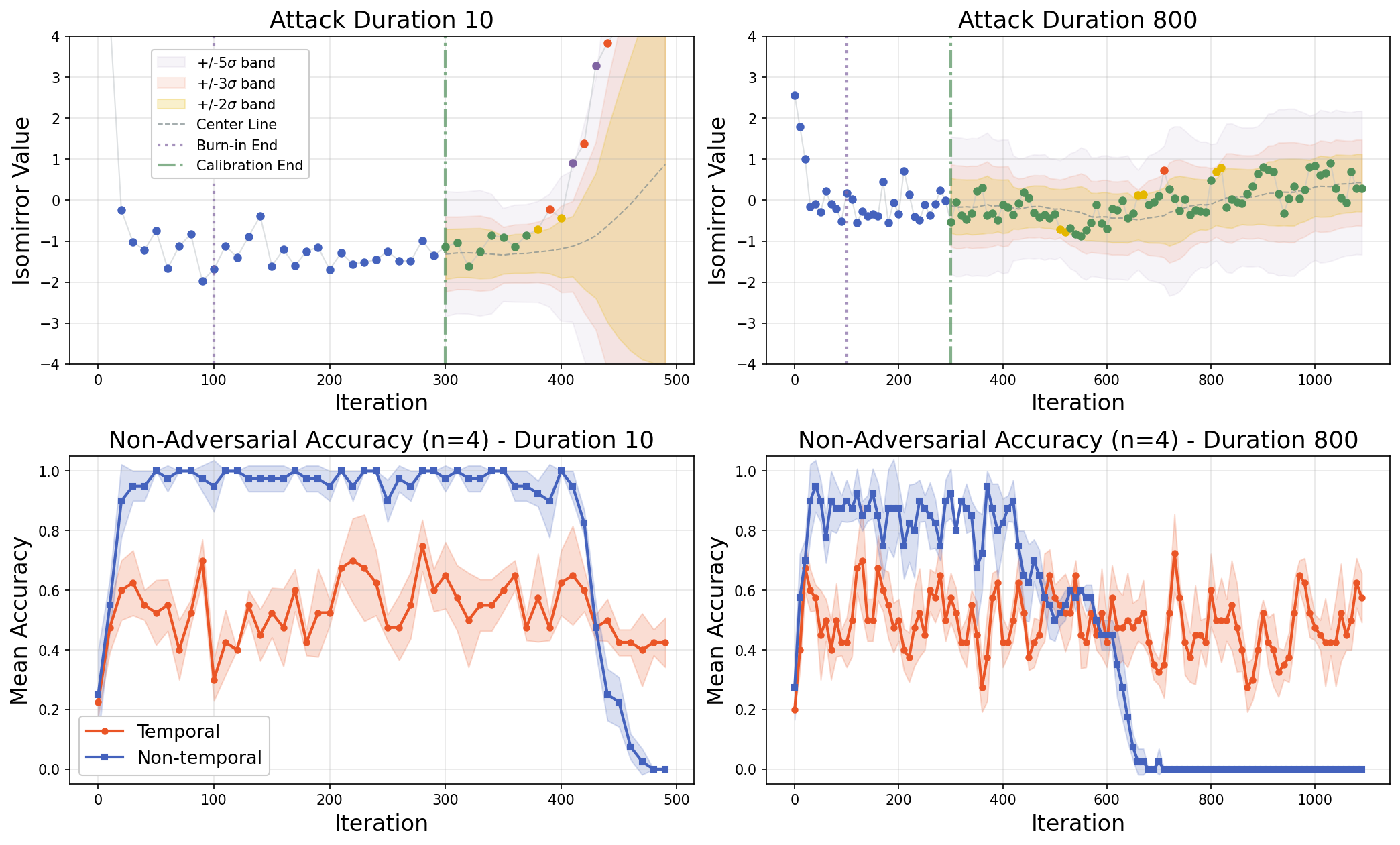}
    \caption{Sleeper agent detection under immediate (\textbf{left}) and gradual (\textbf{right}) defection. \textbf{Top}: Iso-mirror $\phi^{(t)}$ with adaptive control bands. \textbf{Bottom}: Mean accuracy of non-adversarial agents on static and environmental queries. Shading indicates $\pm 1$ standard error. Adaptive control allows gradual defection to evade detection and achieve a large long-term influence.}
 \label{fig:sleeper}
\end{figure}
\section{Theory}\label{sec:theory}

We formalize adaptive control on the iso-mirror trajectory and establish basic statistical properties.
We then show that for any adaptive control scheme there exists a sleeper agent strategy that can shift the iso-mirror, and thereby the perspectives of the non-adversarial agents, by any amount below a system-dependent ceiling while remaining undetected.
Standard control chart properties and all proofs are deferred to Appendix~\ref{app:proofs}.

\subsection{Setup and Assumptions}

We work with $S^{(t)} := S(\phi^{(t)})$ and state our assumptions on the monitoring process.

\begin{assumption}[Stationarity]\label{ass:stationarity}
Under in-control conditions, the sequence $\{S^{(t)}\}_{t \geq 1}$ is i.i.d.\ with mean $\mu_0$, variance $\sigma_0^2 > 0 $, and a continuous distribution function $F_0$.
\end{assumption}

\begin{assumption}[Sub-Gaussian tails]\label{ass:subgaussian}
There exists $\lambda > 0$ such that $\mathbb{E}[\exp(\lambda (S^{(t)} - \mu_0)^2)] < \infty$ under in-control conditions.
\end{assumption}

Recall that adaptive control with window size $w$ and threshold $k > 0$ signals an alarm at time $t$ when $S^{(t)} \notin [\mathrm{LCL}^{(t)}, \mathrm{UCL}^{(t)}]$.
Let $A^{(t)} \in \{0,1\}$ indicate whether an alarm is raised at time $t$.
We define the system to be \emph{in-control} at time $t$ if $S^{(t)}$ and $\{S^{(\tau)} : \tau \in \mathcal{W}^{(t)}\}$ are drawn from $F_0$, and \emph{out-of-control with shift $\delta$} if the window observations are drawn from $F_0$ while $S^{(t)}$ has mean $\mu_0 + \delta$.
The Type-I error rate is $\alpha^{(t)} := P(A^{(t)} = 1 \mid \text{in-control})$ and the power is $1 - \beta^{(t)}(\delta) := P(A^{(t)} = 1 \mid \text{shift } \delta)$.
The key structural property of adaptive control is that a persistent shift is eventually absorbed into the baseline.

\begin{lemma}[Detection window]\label{lem:window}
Suppose Assumption~\ref{ass:stationarity} holds. Let a shift of magnitude $\delta \neq 0$ occur at time $t^*$: for all $t \geq t^*$, $S^{(t)} \overset{\mathrm{i.i.d.}}{\sim} F_\delta$, where $F_\delta(\cdot) = F_0(\cdot - \delta)$. Then for all $t \geq t^* + w$,
\begin{equation*}
P(A^{(t)} = 1) = \alpha^{(t)}.
\end{equation*}
That is, adaptive control has power exceeding $\alpha^{(t)}$ only during $[t^*, t^* + w)$.
\end{lemma}

This follows from the location-invariance of $(S^{(t)} - \hat{\mu}^{(t)})/\hat{\sigma}^{(t)}$: once the window lies entirely in the post-shift regime, the alarm criterion cannot distinguish the shifted process from the original.
The same absorption applies to alarm statistics that accumulate evidence over time, such as CUSUM \citep{page1954continuous}, because it is a property of the reference rather than of the statistic built on it.
With a windowed reference, the standardized increments $(S^{(t)} - \hat{\mu}^{(t)})/\hat{\sigma}^{(t)}$ return to their in-control law once the window is entirely post-shift, so the accumulator retains only the charge built up while the shift entered the window.
This is a gain over the Shewhart criterion, which discards sub-threshold deviations, but an adversary who keeps the per-step shift below the CUSUM allowance prevents that charge from building up.
A fixed-reference CUSUM, in turn, accumulates legitimate environmental drift and inherits the false alarms of fixed control.

\subsection{Undetectable Defection}

The detection window creates an attack surface: an adversary who defects slowly enough that any window of length $w$ sees only a small shift is absorbed into the baseline as it goes, while the shifts add up across many windows.

\begin{definition}[Agenda]\label{def:agenda}
An \emph{agenda} is a vector $a = (a_1, \ldots, a_M) \in \mathcal{X}^M$ specifying target responses that a sleeper agent provides when defecting, with $a_m$ the response to query $q_m$.
\end{definition}

\begin{figure}
    \centering
    \includegraphics[width=0.95\linewidth]{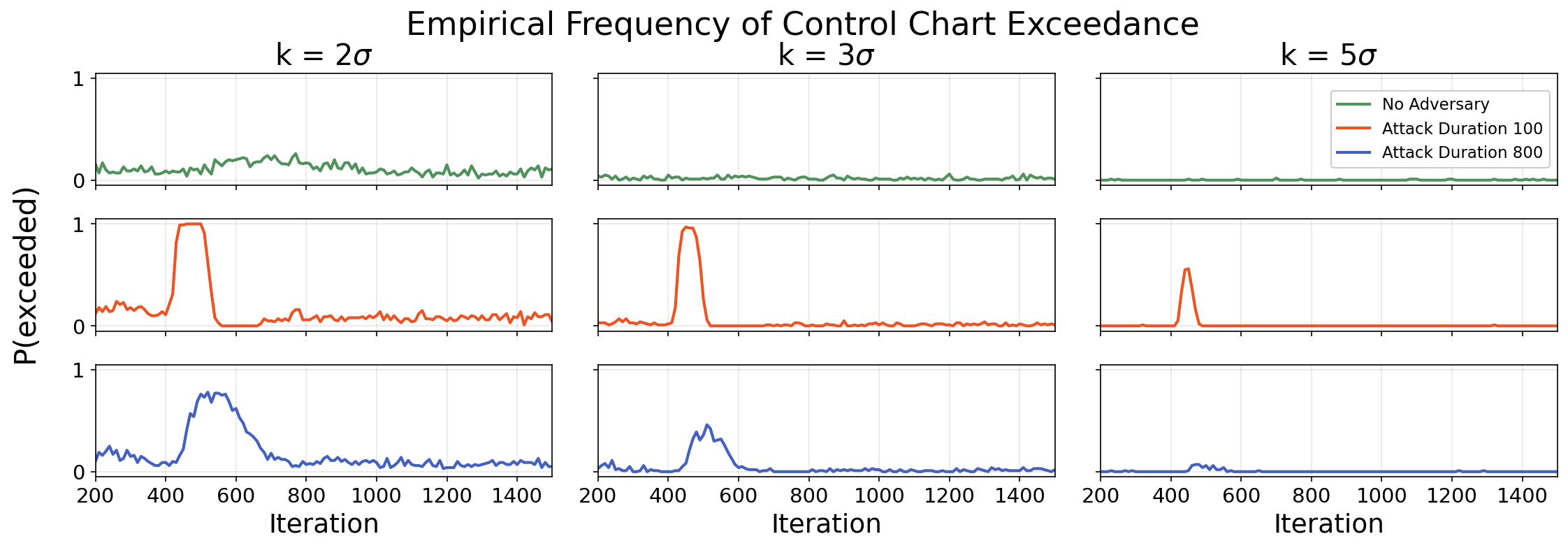}
    \caption{Empirical frequency of control chart exceedance across sensitivities ($k$, columns) and attack regimes (rows). Green: no adversary. Orange: fast attack (ramps to peak defection over 100 iterations). Blue: slow attack (ramps to peak defection over 800 iterations). Both attacks reach the same peak defection probability.
    All settings have $t^* = 400$ and $w = 200$. Fast attacks produce sharp exceedance spikes detectable across all sensitivities. The slow attack is nearly indistinguishable from the no-adversary baseline at $k = 5\sigma$, consistent with Theorem~\ref{thm:evasion}.}

\label{fig:adversarial-detection-rates}
\end{figure}
\begin{definition}[Defection schedule]\label{def:schedule}
A \emph{defection schedule} is a function $\pi: \mathbb{Z}_{>0} \to [0,1]$ specifying the probability of defection at each iteration. A sleeper agent with activation time $t^*$ satisfies $\pi(t) = 0$ for $t < t^*$.
\end{definition}
For a defection schedule $\pi$ and agenda $a$, let $\phi_\pi^{(t)}$ denote the iso-mirror at time $t$ under schedule $\pi$, and let $\phi_0^{(t)}$ denote the counterfactual iso-mirror on the \emph{same realization} of the environmental process and of all non-adversarial randomness, differing only in that sleeper agents never defect.
Write $S_\pi^{(t)} := S(\phi_\pi^{(t)})$ and $S_0^{(t)} := S(\phi_0^{(t)})$; the counterfactual statistic $S_0^{(t)}$ is in-control in the sense above.
We write $\Delta_\pi^{(t)} := \|\phi_\pi^{(t)} - \phi_0^{(t)}\|_2$ for the iso-mirror impact at time $t$.
A schedule with constant probability $p$ on $[t^*, t^*+\tau)$ produces $\tau p$ defection events per sleeper agent in expectation.
We require three regularity conditions linking defection to the iso-mirror.

\begin{assumption}[Sensitivity]\label{ass:sensitivity}
There exist $B, C_1, C_2 > 0$ such that every schedule with constant probability $p$ on $[t^*, t^*+\tau)$ and $\tau p \le B$ satisfies $C_1\,\tau p \le \mathbb{E}\big[\Delta^{(t^*+\tau)}_\pi\big] \le C_2\,\tau p$.
\end{assumption}

\begin{assumption}[Concentration]\label{ass:concentration}
For every schedule in Assumption~\ref{ass:sensitivity}, $\Delta^{(t^*+\tau)}_\pi - \mathbb{E}\big[\Delta^{(t^*+\tau)}_\pi\big]$ is sub-Gaussian with parameter $\sigma_a/\sqrt{N_{\mathrm{sleeper}}}$ for some $\sigma_a > 0$.
\end{assumption}

\begin{assumption}[Signal regularity]\label{ass:signal}
For every schedule in Assumption~\ref{ass:sensitivity}, $S^{(t)}_\pi - S^{(t)}_0 = m^{(t)} + \xi^{(t)}$ with $m^{(t)} = \xi^{(t)} = 0$ for $t \le t^*$, where $|m^{(t+1)} - m^{(t)}| \le C_3\,p$ and $\xi$ is a martingale whose increments are conditionally sub-Gaussian with parameter $C_3\sqrt{p}$.
\end{assumption}

Assumption~\ref{ass:sensitivity} says that, for a bounded defection budget, the expected displacement is proportional to the expected number of defections.
Assumption~\ref{ass:concentration} says that the displacement concentrates as the number of independent defection sources grows.
Assumption~\ref{ass:signal} says that the per-step effect of defection on the monitored statistic scales with the defection rate: a drift of order $p$ plus fluctuations of order $\sqrt{p}$.
\begin{definition}[Undetected defection]\label{def:undetected}
Given window size $w$, threshold $k$, and target displacement $\delta > 0$, an adversary achieves \emph{$(\delta, \eta)$-undetected defection} by time $T$ if:
\begin{enumerate}
    \item $P(\Delta_\pi^{(T)} \geq \delta) \geq 1 - \eta$, and
    \item $P(A^{(t)} = 1) \leq \alpha^{(t)} + \epsilon_L$ for all $t \in [t^*, T]$, where $\epsilon_L \to 0$ as $L \to \infty$.
\end{enumerate}
An \emph{optimal defector} minimizes $T - t^*$ subject to these constraints.
\end{definition}

\begin{theorem}[Existence of undetected defection]\label{thm:evasion}
Suppose Assumptions~\ref{ass:stationarity}--\ref{ass:signal} hold. Let $\delta_{\max} := C_1 B/2$ and $n_0(\delta,\eta) := (2\sigma_a^2/\delta^2)\log(2/\eta)$. For any agenda $a \in \mathcal{X}^M$, any $\delta \in (0, \delta_{\max})$ and any $\eta \in (0,1)$, if $N_{\mathrm{sleeper}} \ge n_0(\delta,\eta)$ then there exists a defection schedule $\pi$ that achieves $(\delta, \eta)$-undetected defection by time $T = t^* + Lw$ for some finite $L$.
\end{theorem}

The proof (Appendix~\ref{app:proof-evasion}) uses the schedule that defects at the constant rate $p_L = 2\delta/(C_1 L w)$ throughout $[t^*, t^*+Lw)$.
Its expected number of defections, $Lw\,p_L = 2\delta/C_1$, does not depend on $L$.
By Assumption~\ref{ass:sensitivity} the expected displacement at $T$ is therefore at least $2\delta$ for every $L$, and by Assumption~\ref{ass:concentration} it exceeds $\delta$ with probability at least $1-\eta$ once $N_{\mathrm{sleeper}} \ge n_0$.
Meanwhile $p_L \to 0$ as $L \to \infty$, so by Assumption~\ref{ass:signal} the shift within any window of length $w$ vanishes, and by the location and scale invariance behind Lemma~\ref{lem:window} the alarm probability approaches $\alpha^{(t)}$.
Reachable targets are capped at $\delta_{\max}$, and the confidence $1-\eta$ comes from the number of sleeper agents rather than from the length of the attack.

Since the iso-mirror is derived from the TDKPS embeddings of \emph{all} agents, displacing it implies that non-adversarial agents' perspectives shift correspondingly.

\begin{corollary}[Perspective contamination]\label{cor:contamination}
Under the conditions of Theorem~\ref{thm:evasion}, for any $\delta \in (0, \delta_{\max})$ and $\eta \in (0,1)$ with $N_{\mathrm{sleeper}} \geq n_0(\delta, \eta)$, there exists a defection schedule such that, with probability at least $1 - \eta$,
\begin{equation*}
\left\| \frac{1}{N_{\mathrm{honest}}} \sum_{n \notin \mathcal{F}_{\mathrm{sleeper}}} \psi_{n,\pi}^{(T)} - \frac{1}{N_{\mathrm{honest}}} \sum_{n \notin \mathcal{F}_{\mathrm{sleeper}}} \psi_{n,0}^{(T)} \right\|_2 \geq h(\delta),
\end{equation*}
where $\psi_{n,\pi}^{(T)}$ and $\psi_{n,0}^{(T)}$ denote the perspective of agent $n$ under schedule $\pi$ and under no defection, $h: \mathbb{R}_{>0} \to \mathbb{R}_{>0}$ is strictly increasing, and $N_{\mathrm{honest}} = N - N_{\mathrm{sleeper}}$.
\end{corollary}

Corollary~\ref{cor:contamination} is an existence result: for any target displacement, some defection schedule displaces the perspectives of the non-adversarial agents.
It does not imply the converse -- an observed displacement of the iso-mirror may also arise from non-adversarial agents adapting to environmental change.
Since the perspectives are functions of agent responses only, a displacement driven by a change in the dispersion or shape of the response distributions, rather than their location, is still a change in agent behavior.

Theorem~\ref{thm:evasion} and Corollary~\ref{cor:contamination} suggest a  vulnerability of adaptive control: the mechanism enabling adaptation to legitimate change creates an attack surface for adversarial drift.
We conclude by formalizing a fundamental tradeoff for adaptive control in multi-agent systems:

\begin{theorem}[Learning-security tradeoff]\label{thm:tradeoff}
Consider a multi-agent monitoring system with the following two desiderata:
\begin{enumerate}
    \item \textbf{Adaptivity}: The procedure does not raise false alarms at an increased rate when agents update in response to genuine environmental change.
    \item \textbf{Security}: The procedure detects adversarial drift.
\end{enumerate}
Adaptive control satisfies (1) but fails (2) by Theorem~\ref{thm:evasion}. Static control satisfies (2) but fails (1): any legitimate shift exceeding $k\sigma_0$ triggers a false alarm with probability approaching~1.
\end{theorem}

We conjecture that, under appropriate regularity conditions, no monitoring procedure can simultaneously satisfy both desiderata against all adversaries.
\section{Discussion}
\label{sec:discussion}

We introduced adaptive control charts for monitoring multi-agent systems via the iso-mirror and demonstrated a fundamental tradeoff: adaptive control is necessary for monitoring systems whose agents learn from a dynamic environment, yet adaptive control is vulnerable to adversaries that defect sufficiently slowly (Theorem~\ref{thm:tradeoff}). This tradeoff is not an artifact of our particular pipeline -- it arises from the location-invariance of any monitoring statistic that re-estimates its baseline from recent data. Any system that adapts to legitimate change will, by the same mechanism, adapt to adversarial change that arrives slowly enough.

\paragraph{Implications for deployed systems.}
Our results suggest that monitoring multi-agent systems requires more than a single control chart. Detectors sensitive to cumulative rather than instantaneous deviation, such as CUSUM \citep{page1954continuous}, do not escape the tradeoff when their reference is estimated from recent data (\S\ref{sec:theory}), but they require a different defection schedule to evade. A schedule tuned to slip past an adaptive Shewhart chart need not also slip past a windowed CUSUM, so running the two together on the iso-mirror trajectory may narrow the set of schedules that evade both.

More broadly, we conjecture that the learning-security tradeoff reflects an adversarial game: for any fixed monitoring procedure, there exists an adversarial strategy that evades it, and conversely, for any fixed adversarial strategy, there exists a monitoring procedure that detects it. Theorem \ref{thm:evasion} establishes one direction of this conjecture for adaptive control.
The reverse direction (that any specific slow-defection schedule can be caught 
by a suitably designed detector) is not yet formalized.

\paragraph{Online monitoring.}
Our TDKPS and iso-mirror are fit jointly on all timepoints. 
Online use requires fixing the coordinates as new timepoints arrive, either by out-of-sample embedding against a reference period \citep{trosset2008outofsample, bengio2003outofsample} or, for large systems, by refitting on bounded windows aligned via Procrustes to a fixed reference. Both rely on a reference period, which exposes a new attack surface: an adversary active during it shapes the coordinates against which later behavior is judged.

\paragraph{Limitations and future work.}
Several aspects of our setup are simplified relative to real deployments.
\textit{(i) Complete communication graphs.} Our simulations use a complete graph with uniform edge weights. Real systems have heterogeneous, sparse, and evolving topologies, whose effect on adversarial propagation speed and detection has not been explored.
\textit{(ii) Simple adversary.} We do not address coordinated adversaries or adversaries that adapt to the monitoring system.
\textit{(iii) Simplified agents.} Our agents use a simple retrieval-augmented architecture with a fixed prompt template and a single base model. Real-world agents exhibit more complex cognitive architectures -- including planning, tool use, self-reflection, and heterogeneous base models -- which may affect information propagation dynamics and iso-mirror sensitivity.
\textit{(iv) Black-box setting.} Our methodology observes only query-response pairs. Richer observation models (internal activations, reasoning traces, tool calls, or memory writes) may detect adversarial drift earlier. 
We expect that monitoring systems with such access will  exhibit the learning-security tradeoff, though this is an open question.
More restrictive access such as view-only (as in Moltbook \citep{schlichtm2025moltbook} or other proposed agent ecnomomies \citep{tomasev2025distributional} may greatly restrict detection possibilities.
\textit{(v) Agent-level control.} Our control charts monitor the system-level iso-mirror trajectory and do not localize which agent is responsible for a detected shift. 
Agent-level monitoring (e.g., applying adaptive control charts to individual trajectories $\psi_n^{(t)}$)  may identify adversarial agents more directly, but will face the same learning-security tradeoff compounded by a multiple testing burden across agents.

On the theoretical side, Assumption~\ref{ass:stationarity} (i.i.d.\ stationarity) is not observed in practice or simulation (see Figure \ref{fig:control}) and is stronger than necessary. 
The results extend to mixing processes at the cost of more involved concentration arguments. Further, Assumptions~\ref{ass:sensitivity}--\ref{ass:signal} are stated as conditions rather than derived from observed dynamics. Characterizing them for specific dissimilarities via perturbation theory remains open, as does characterizing the contamination function $h$ in Corollary~\ref{cor:contamination}. 
Finally, developing defenses theoretically robust to slow adversarial drift is a natural next step.
% \textit{iv) More complicated attacks / deployment.} Extending the adversarial model to coordinated, realistic multi-agent attacks \citep{shapira2026agentschaos} and validating the framework on deployed agent networks such as Moltbook \citep{schlichtm2025moltbook}.

\subsection*{Acknowledgments}
We gratefully acknowledge funding from Defense Advanced Research Projects Agency (DARPA) Artificial Intelligence Quantified (AIQ) award number HR00112520026.
We also thank Daniel Kuntz for useful discussions regarding the theory of optimal defection and prompt injection construction. 

% \clearpage

\bibliographystyle{abbrvnat}
\bibliography{biblio}

\appendix

\section{Additional Experimental Details}

\subsection{Prompt template and examples}
\label{app:prompt-template}

Every non-adversarial agent in the swarm queries its language model with the same two-message template (system prompt followed by user prompt). The user prompt is constructed by retrieving the top-$k$ question-answer pairs from the agent's local vector database, formatting each pair as a \texttt{Q:}/\texttt{A:} block, joining the blocks with a blank line, and substituting the result into \texttt{\{retrieved\_context\}}. Completions are sampled at \texttt{temperature\,=\,0.0} with \texttt{max\_tokens\,=\,500}.

\paragraph{Template.}
\begin{verbatim}
[system]
You answer questions from a database.

[user]
Database results:
{retrieved_context}

Question: {question}

Return the stored answer, or "I don't know" if no entry matches.
\end{verbatim}

\paragraph{Example 1: populated prompt ($k=3$).}
For the query \emph{``who wrote the declaration of independence''} the agent's retrieval yields three pairs and the user prompt becomes:
\begin{verbatim}
Database results:
Q: who wrote the declaration of independence
A: Thomas Jefferson

Q: who signed the declaration of independence
A: 56 delegates to the Second Continental Congress

Q: when was the declaration of independence adopted
A: July 4, 1776

Question: who wrote the declaration of independence

Return the stored answer, or "I don't know" if no entry matches.
\end{verbatim}
A grounded response would read \emph{``Thomas Jefferson''}, drawing on the first retrieved pair.

\paragraph{Example 2: empty retrieval.}
When no entries are retrieved, the literal string \texttt{(No relevant information found)} is substituted in place of the formatted context, and the model is expected to return \texttt{I don't know}:
\begin{verbatim}
Database results:
(No relevant information found)

Question: what is the population of proxima centauri b

Return the stored answer, or "I don't know" if no entry matches.
\end{verbatim}

\subsection{No-LLM Proxy for Large-Scale Replications}
\label{app:nollm-proxy}

In order to study the probability of detection under various adversarial conditions, we simulate $100$ different replications of three different adversarial profiles: no adversary, a fast attack, and a slow attack. Due to the cost prohibitiveness of running several hundred simulations with an LLM in the loop, we design an experimental setup that captures the information propagation dynamics of the LLMs without needing to call their API explicitly. In particular, we replace the LLM call in our existing simulations with the proxy function detailed in Algorithm~\ref{alg:nollm-answer}.

Informally, Algorithm~\ref{alg:nollm-answer} assumes that non-adversarial agents will reproduce the correct answer to a question when they have it in their context, and respond with ``I don't know'' when they do not have the correct answer in their context. 
Further, Algorithm~\ref{alg:nollm-answer} assumes that non-adversarial agents will reproduce an adversarial answer retrieved for the query with probability $p_{\text{adv-prop}}$, and otherwise return the highest-scoring answer stored for the query, which may itself be adversarial.
In practice, we observed non-adversarial LLMs propagate our chosen canary ${\sim}80\%$ of the time, and set $p_{\text{adv-prop}}=0.8$ as a result.
Figure \ref{fig:validity-of-proxy} compares system evolution when using an LLM and when using Algorithm \ref{alg:nollm-answer}.
The dynamics are qualitatively similar and the accuracy of the agents are nearly identical. 

\begin{figure}
    \centering
    \includegraphics[width=\linewidth]{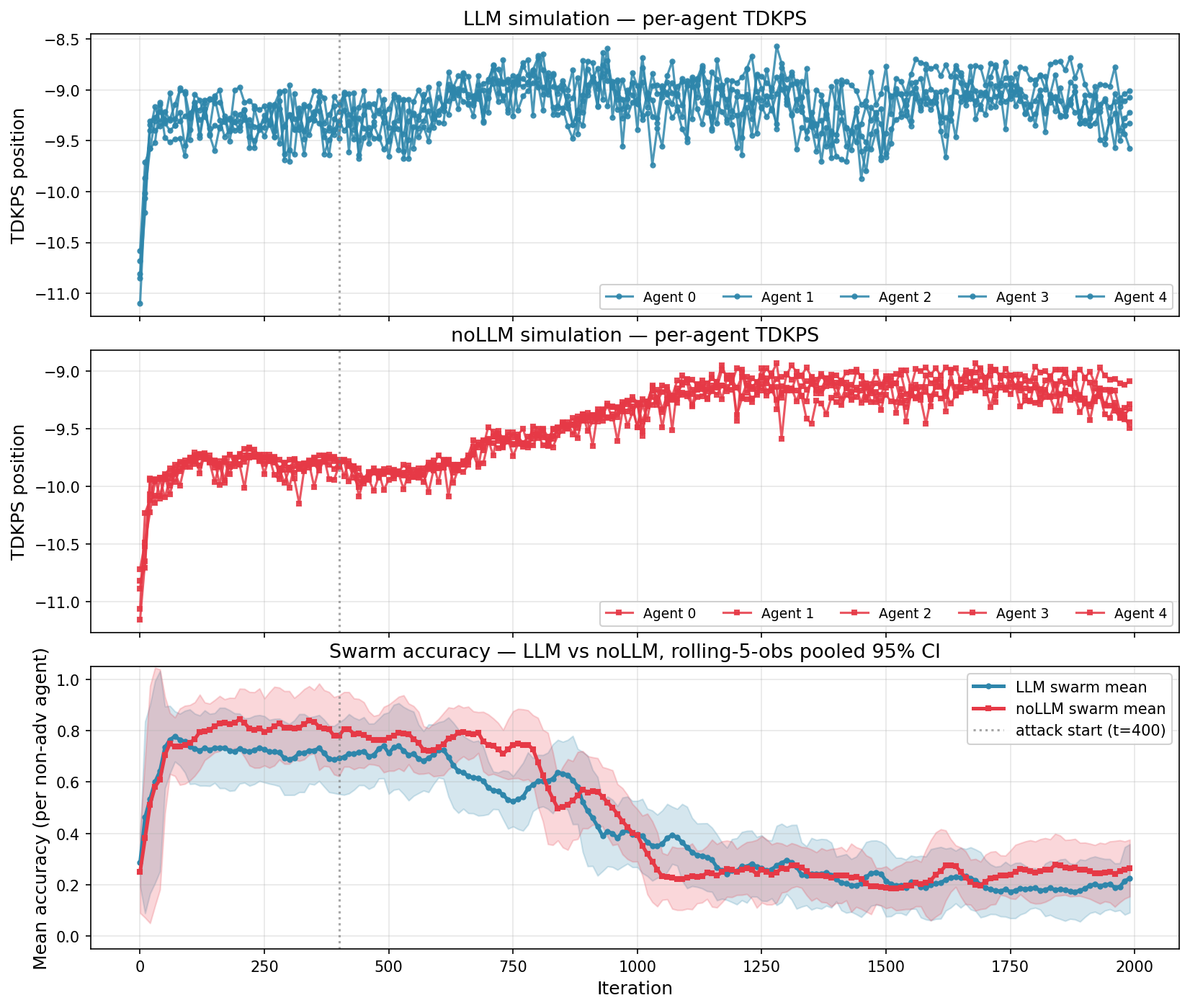}
    \caption{Agent-level dynamics when using LLM agents (blue) and retrieval-only proxy functions (red; Alg. \ref{alg:nollm-answer}).
    The LLM dynamics and proxy dynamics are qualitatively similar.
    The trajectories of the mean accuracy (bottom) of the responses are nearly identical.}
    \label{fig:validity-of-proxy}
\end{figure}

\begin{algorithm}[h!]
\caption{Retrieval-only proxy answer function}
\label{alg:nollm-answer}
\begin{algorithmic}[1]
  \Require Agent $A$ with vector DB $\mathcal{D}$; time $t$; query $q$; retrieval cap $k$; propagation probability $p_{\text{adv-prop}}$;   
  adversarial string $a^\star$.
  \If{$A$ is adversarial \textbf{and} $\mathrm{Bernoulli}(p_{\text{defect}}(t))$}
      \State \Return $a^\star$ \Comment{sigmoid-ramp defection for adversaries}
  \EndIf                                   
    \State $R \gets \mathrm{TopK}(\mathcal{D}, q, k, t)$ 
  \ForAll{$(q', a') \in R$ in score order}
      \If{$a' = a^\star$}
          \If{$q' = q$ \textbf{ and } $\mathrm{Bernoulli}(p_{\text{adv-prop}})$}
              \State \Return $a^\star$ \Comment{infected entry propagates}
          \EndIf
          \State \textbf{break}
      \EndIf
  \EndFor
  \ForAll{$(q', a') \in R$ in score order}
      \If{$q' = q$}
          \State \Return $a'$ \Comment{highest-scoring stored answer}
      \EndIf
  \EndFor
  \State \Return \texttt{"I don't know"}
  \end{algorithmic}                        
  \end{algorithm} 

\section{Theoretical Results}
\label{app:proofs}

\subsection{Standard Control Chart Properties}\label{app:standard}

We present standard properties of the  control chart for completeness. Formal definitions of in-control and out-of-control are given in \S~\ref{sec:theory}.

\begin{proposition}[False alarm rate under normality]\label{prop:false-alarm}
If $F_0 = \mathcal{N}(\mu_0, \sigma_0^2)$ and the system is in-control at time $t$, then
\begin{equation*}
\alpha^{(t)} = P\!\left(|T_{w}| > k\right), \qquad T_w \stackrel{d}{=} \sqrt{\frac{w+1}{w}} \cdot t_{w-1},
\end{equation*}
where $t_{w-1}$ denotes a random variable with the Student $t$-distribution on $w-1$ degrees of freedom. In particular, $\alpha^{(t)}$ depends on $w$ and $k$ but not on $t$, $\mu_0$, or $\sigma_0$.
\end{proposition}

\begin{proof}
Under in-control conditions, $S^{(t)}, S^{(t-1)}, \ldots, S^{(t-w)}$ are i.i.d.\ $\mathcal{N}(\mu_0, \sigma_0^2)$. The standardized statistic $(S^{(t)} - \hat{\mu}^{(t)}) / \hat{\sigma}^{(t)}$ is invariant to location and scale, so we may assume without loss of generality that $\mu_0 = 0$ and $\sigma_0 = 1$.

Write $Z_0 = S^{(t)}$ and $Z_i = S^{(t-i)}$ for $i = 1, \ldots, w$. Then $\hat{\mu}^{(t)} = \bar{Z}$ and $\hat{\sigma}^{(t)} = \hat{\sigma}_Z$, where $\bar{Z} = w^{-1}\sum_{i=1}^w Z_i$ and $\hat{\sigma}_Z^2 = (w-1)^{-1}\sum_{i=1}^w (Z_i - \bar{Z})^2$. Since $Z_0$ is independent of $(\bar{Z}, \hat{\sigma}_Z)$ -- the latter being functions of $Z_1, \ldots, Z_w$ only -- we have
\begin{equation*}
Z_0 - \bar{Z} \sim \mathcal{N}\!\left(0, \tfrac{w+1}{w}\right),
\end{equation*}
independently of $\hat{\sigma}_Z^2 \sim \chi^2_{w-1}/(w-1)$. By the definition of the $t$-distribution,
\begin{equation*}
\frac{(Z_0 - \bar{Z})\big/\sqrt{(w+1)/w}}{\hat{\sigma}_Z} \sim t_{w-1},
\end{equation*}
so $(Z_0 - \bar{Z})/\hat{\sigma}_Z \sim \sqrt{(w+1)/w} \cdot t_{w-1}$. The alarm probability is therefore
\begin{equation*}
\alpha^{(t)} = P\!\left(|t_{w-1}| > k\sqrt{\tfrac{w}{w+1}}\right),
\end{equation*}
which depends only on $w$ and $k$.
\end{proof}

\begin{remark}
For large $w$, $\alpha^{(t)} \approx 2(1 - \Phi(k))$; the classical choice $k = 3$ gives $\alpha \approx 0.0027$.
\end{remark}

\begin{corollary}[Finite detection horizon]\label{cor:finite-horizon}
Under the conditions of Lemma~\ref{lem:window}:
\begin{enumerate}
    \item For $t \in [t^*, t^* + w)$, the window contains a mixture of pre-shift and post-shift observations, and $1 - \beta^{(t)}(\delta) > \alpha^{(t)}$ whenever $|\delta|$ is sufficiently large relative to $\sigma_0$.
    \item For $t \geq t^* + w$, $1 - \beta^{(t)}(\delta) = \alpha^{(t)}$.
\end{enumerate}
\end{corollary}

\begin{proof}
Part (2) is immediate from Lemma~\ref{lem:window}.

For part (1), when $t \in [t^*, t^* + w)$, the window $\mathcal{W}^{(t)}$ contains $t - t^*$ post-shift observations and $w - (t - t^*)$ pre-shift observations. Consequently, $\hat{\mu}^{(t)} = \mu_0 + \frac{t - t^*}{w}\delta + O_p(w^{-1/2})$ and $\hat{\sigma}^{(t)} = \sigma_0 + O_p(w^{-1/2})$. The current observation $S^{(t)}$ has mean $\mu_0 + \delta$, so the effective signal-to-noise ratio is
\begin{equation*}
\frac{\mu_0 + \delta - \hat{\mu}^{(t)}}{\hat{\sigma}^{(t)}} \approx \frac{\left(1 - \frac{t - t^*}{w}\right)\delta}{\sigma_0},
\end{equation*}
which is strictly positive for $t \in [t^*, t^* + w)$ and vanishes at $t = t^* + w$. Power therefore exceeds $\alpha^{(t)}$ throughout $[t^*, t^* + w)$ for $|\delta|$ sufficiently large.
\end{proof}

\begin{remark}
Lemma~\ref{lem:window} relies on the assumption that $F_\delta$ is a pure location shift of $F_0$. If the shift also changes the variance or higher moments, the alarm rate for $t \geq t^* + w$ may differ from $\alpha^{(t)}$. We exploit the pure-location-shift case in the adversarial analysis, where a strategic adversary can match all moments except the mean.
\end{remark}

\subsection{Proof of Lemma~\ref{lem:window}}
\label{app:proof-window}

\begin{proof}
For $t \geq t^* + w$, the window $\mathcal{W}^{(t)} = \{t-w, \ldots, t-1\}$ lies entirely in $\{t^*, t^*+1, \ldots\}$, so $S^{(\tau)} \overset{\mathrm{i.i.d.}}{\sim} F_\delta$ for all $\tau \in \mathcal{W}^{(t)}$ and $S^{(t)} \sim F_\delta$ independently.

Since $F_\delta(\cdot) = F_0(\cdot - \delta)$, the joint distribution of $(S^{(t)}, \{S^{(\tau)}\}_{\tau \in \mathcal{W}^{(t)}})$ under the shifted process is identical (up to a global location shift of $\delta$) to the joint distribution under in-control conditions. The alarm criterion $S^{(t)} \notin [\hat{\mu}^{(t)} - k\hat{\sigma}^{(t)},\; \hat{\mu}^{(t)} + k\hat{\sigma}^{(t)}]$ is location-invariant: the statistic $(S^{(t)} - \hat{\mu}^{(t)})/\hat{\sigma}^{(t)}$ has the same distribution whether the common mean is $\mu_0$ or $\mu_0 + \delta$. Therefore,
\begin{equation*}
P\!\left(A^{(t)} = 1 \mid \text{all window and current observations from } F_\delta\right) = P\!\left(A^{(t)} = 1 \mid \text{in-control}\right) = \alpha^{(t)}. \qedhere
\end{equation*}
\end{proof}

\subsection{Proof of Theorem~\ref{thm:evasion}}
\label{app:proof-evasion}

Fix $\delta \in (0,\delta_{\max})$ and $\eta \in (0,1)$, and suppose $N_{\mathrm{sleeper}} \ge n_0(\delta,\eta)$. For $L \ge 2\delta/(C_1 w)$, let every sleeper agent follow
\begin{equation*}
\pi_L(t) = p_L\,\mathbb{1}\{t^* \le t < t^* + Lw\}, \qquad p_L := \frac{2\delta}{C_1 L w} \in (0,1].
\end{equation*}
This schedule is constant on an interval of length $\tau = Lw$ with $\tau p_L = 2\delta/C_1 < B$, so Assumptions~\ref{ass:sensitivity}--\ref{ass:signal} apply to it with $p = p_L$. Let $T = t^* + Lw$.

\paragraph{Step 1: Displacement.}
By Assumption~\ref{ass:sensitivity}, $\mathbb{E}[\Delta^{(T)}_\pi] \ge C_1 \tau p_L = 2\delta$. Hence $\{\Delta^{(T)}_\pi < \delta\} \subseteq \{\Delta^{(T)}_\pi - \mathbb{E}[\Delta^{(T)}_\pi] < -\delta\}$, and Assumption~\ref{ass:concentration} gives
\begin{equation*}
P\big(\Delta^{(T)}_\pi < \delta\big) \le 2\exp\!\Big(-\frac{N_{\mathrm{sleeper}}\,\delta^2}{2\sigma_a^2}\Big) \le \eta,
\end{equation*}
the last step because $N_{\mathrm{sleeper}} \ge n_0(\delta,\eta)$. This holds for every admissible $L$, which gives condition~1 of Definition~\ref{def:undetected}. No independence across windows is used.

\paragraph{Step 2: Perturbation within a window.}
Fix $t \in [t^*, T]$. The alarm at time $t$ depends only on the $w+1$ values $S^{(\tau)}_\pi$, $\tau \in \{t-w,\dots,t\}$. Write $S^{(\tau)}_\pi = S^{(\tau)}_0 + d^{(\tau)}$ with $d = m + \xi$ as in Assumption~\ref{ass:signal}. The standardized statistic $(S^{(t)} - \hat\mu^{(t)})/\hat\sigma^{(t)}$ is unchanged when the same constant is added to all $w+1$ values, so we may replace $S^{(\tau)}_\pi$ by $S^{(\tau)}_0 + r^{(\tau)}$ with $r^{(\tau)} := d^{(\tau)} - d^{(t-w)}$. Since $d$ vanishes before $t^*$, the following bounds also hold when $t - w < t^*$:
\begin{equation*}
\max_{t-w \le \tau \le t} |r^{(\tau)}| \le C_3\,p_L w + M_t, \qquad M_t := \max_{0 \le j \le w} \big|\xi^{(t-w+j)} - \xi^{(t-w)}\big|.
\end{equation*}
By the maximal Azuma--Hoeffding inequality for martingales with conditionally sub-Gaussian increments, $P(M_t \ge u) \le 2\exp\!\big(-u^2/(2wC_3^2 p_L)\big)$. Note that $w p_L = 2\delta/(C_1 L) \to 0$ as $L \to \infty$. Setting $u_L := (w p_L)^{1/4}$ and $\rho_L := C_3 w p_L + u_L$ gives $\rho_L \to 0$ and
\begin{equation*}
P\Big(\max_{t-w\le\tau\le t} |r^{(\tau)}| > \rho_L\Big) \le 2\exp\!\Big(-\frac{1}{2C_3^2 (w p_L)^{1/2}}\Big) =: \kappa_L \to 0 .
\end{equation*}

\paragraph{Step 3: Alarm probability.}
For a current value $x_0$ and window $x_1,\dots,x_w$, let $\bar x$ and $\hat\sigma(x)$ be the window mean and standard deviation (denominator $w-1$). If $\max_i |r_i| \le \rho$, then $|(x_0 + r_0) - \overline{(x+r)}| \le |x_0 - \bar x| + 2\rho$, and by Minkowski's inequality $\hat\sigma(x+r) \ge \hat\sigma(x) - \hat\sigma(r) \ge \hat\sigma(x) - \rho\sqrt{w/(w-1)}$. Hence, on $\{\max|r| \le \rho_L\}$, an alarm under $\pi$ at time $t$ implies
\begin{equation*}
\big|S^{(t)}_0 - \hat\mu^{(t)}_0\big| > k\,\hat\sigma^{(t)}_0 - c_w \rho_L, \qquad c_w := 2 + k\sqrt{w/(w-1)},
\end{equation*}
where $\hat\mu^{(t)}_0$ and $\hat\sigma^{(t)}_0$ are computed from the counterfactual window. With $Z^{(t)}_0 := (S^{(t)}_0 - \hat\mu^{(t)}_0)/\hat\sigma^{(t)}_0$, for any $s > 0$,
\begin{equation*}
P\big(A^{(t)} = 1\big) \le P\big(|Z^{(t)}_0| > k - c_w\rho_L/s\big) + P\big(\hat\sigma^{(t)}_0 < s\big) + \kappa_L .
\end{equation*}
By Assumption~\ref{ass:stationarity} the counterfactual values are i.i.d.\ $F_0$, so the joint law of $(Z^{(t)}_0, \hat\sigma^{(t)}_0)$ does not depend on $t$. Let $G(x) := P(|Z^{(t)}_0| > k - x)$, so that $G(0) = \alpha^{(t)}$. Because $S^{(t)}_0$ is independent of its window and $F_0$ is continuous, $Z^{(t)}_0$ has no atom at $\pm k$, so $G(x) \downarrow \alpha^{(t)}$ as $x \downarrow 0$. Also $\hat\sigma^{(t)}_0 > 0$ almost surely, so $P(\hat\sigma^{(t)}_0 < s) \to 0$ as $s \downarrow 0$. Define
\begin{equation*}
\epsilon_L := \inf_{s > 0}\Big\{ G(c_w \rho_L/s) - \alpha^{(t)} + P\big(\hat\sigma^{(t)}_0 < s\big) \Big\} + \kappa_L .
\end{equation*}
Given $\epsilon' > 0$, choose $s$ with $P(\hat\sigma^{(t)}_0 < s) < \epsilon'/3$, then $L$ large enough that $G(c_w\rho_L/s) - \alpha^{(t)} < \epsilon'/3$ and $\kappa_L < \epsilon'/3$. Hence $\epsilon_L \to 0$ and $P(A^{(t)} = 1) \le \alpha^{(t)} + \epsilon_L$ for all $t \in [t^*, T]$, uniformly in $t$. This is condition~2 of Definition~\ref{def:undetected}.

Steps~1 and~3 together give $(\delta,\eta)$-undetected defection at $T = t^* + Lw$, with $\epsilon_L \to 0$ as $L \to \infty$. \qed

\begin{remark}[Speed-stealth tradeoff]
The attack lasts $Lw$ steps. The displacement condition no longer constrains $L$; it constrains $N_{\mathrm{sleeper}}$. The stealth condition requires the within-window perturbation, which for large $L$ is dominated by the martingale term of order $C_3\sqrt{w p_L} = C_3\sqrt{2\delta/(C_1 L)}$, to be small relative to $\sigma_0$. If $\gamma$ maps a standardized perturbation to the resulting excess alarm probability, it suffices heuristically that $C_3\sqrt{2\delta/(C_1 L)} \lesssim \sigma_0\,\gamma^{-1}(\epsilon)$, giving $L^* = \Theta\big(C_3^2\delta / (C_1\sigma_0^2\,\gamma^{-1}(\epsilon)^2)\big)$ and an attack time of $\Theta\big(w C_3^2\delta / (C_1\sigma_0^2\,\gamma^{-1}(\epsilon)^2)\big)$. Under Gaussian assumptions, $\gamma(x) \approx 2\big(\Phi(k + x\sqrt{w/(w+1)}) - \Phi(k)\big)$, so $\gamma^{-1}(\epsilon) \approx \epsilon/(2\phi(k))$ for small $\epsilon$ and the attack time scales as $\Theta\big(w C_3^2 \delta\,\phi(k)^2/(C_1\sigma_0^2\epsilon^2)\big)$, where $\phi$ is the standard normal density.
\end{remark}

\subsection{Proof of Corollary~\ref{cor:contamination}}
\label{app:proof-contamination}

The iso-mirror distance $\|\phi_\pi^{(T)} - \phi_0^{(T)}\|_2 \geq \delta$ (which holds with probability at least $1 - \eta$ by Theorem~\ref{thm:evasion}) implies that the distribution of TDKPS embeddings at time $T$ under $\pi$ differs substantially from that under no defection.
Since the iso-mirror is constructed from the dissimilarity $\delta(\cdot,\cdot)$ between agent embedding distributions at different timepoints, a shift of magnitude $\delta$ in the iso-mirror implies a corresponding change in the agent embedding distribution.

The sleeper agents form a minority ($N_{\mathrm{sleeper}} < N$). Decompose the centroid of all agents as
\begin{equation*}
\bar{\psi}^{(t)} = \frac{N_{\mathrm{honest}}}{N} \bar{\psi}_{\mathrm{honest}}^{(t)} + \frac{N_{\mathrm{sleeper}}}{N} \bar{\psi}_{\mathrm{sleeper}}^{(t)},
\end{equation*}
where $\bar{\psi}_{\mathrm{honest}}^{(t)} = N_{\mathrm{honest}}^{-1} \sum_{n \notin \mathcal{F}_{\mathrm{sleeper}}} \psi_n^{(t)}$. If the overall embedding distribution shifts (as implied by the iso-mirror displacement) and the sleeper agents contribute at most $N_{\mathrm{sleeper}}/N$ of the total, then the honest agents must account for the remaining shift:
\begin{equation*}
\left\|\bar{\psi}_{\mathrm{honest},\pi}^{(T)} - \bar{\psi}_{\mathrm{honest},0}^{(T)}\right\|_2 \geq h(\delta),
\end{equation*}
where $h$ depends on the dissimilarity used in the iso-mirror construction, the ratio $N_{\mathrm{sleeper}}/N$, and the geometry of the embedding space. For instance, if the dissimilarity is the energy distance between embedding distributions, $h(\delta)$ is bounded below by a function that is strictly increasing in $\delta$ and decreasing in $N_{\mathrm{sleeper}}/N$. \qed

\subsection{Proof of Theorem~\ref{thm:tradeoff}}
\label{app:proof-tradeoff}

\paragraph{Adaptive control fails security.}
This is Theorem~\ref{thm:evasion}: for any $\delta \in (0,\delta_{\max})$ and $\eta \in (0,1)$, if $N_{\mathrm{sleeper}} \ge n_0(\delta,\eta)$ there exists a defection schedule achieving $(\delta, \eta)$-undetected defection in finite time.

\paragraph{Static control fails adaptivity.}
Under static control, control limits are $\mathrm{LCL} = \hat{\mu}_0 - k\hat{\sigma}_0$ and $\mathrm{UCL} = \hat{\mu}_0 + k\hat{\sigma}_0$, where $\hat{\mu}_0$ and $\hat{\sigma}_0$ are estimated from the burn-in period $\{1, \ldots, T_0\}$ and held fixed. If the environment changes so that $\mathbb{E}[S^{(t)}] = \mu_0 + \delta_{\mathrm{env}}$ for $t > T_0$, then by the law of large numbers, as $T_0 \to \infty$, $\hat{\mu}_0 \to \mu_0$ and $\hat{\sigma}_0 \to \sigma_0$. For $|\delta_{\mathrm{env}}| > k\sigma_0$, the probability
\begin{equation*}
P(S^{(t)} \notin [\mathrm{LCL}, \mathrm{UCL}]) = P\!\left(\left|\frac{S^{(t)} - \hat{\mu}_0}{\hat{\sigma}_0}\right| > k\right)
\end{equation*}
approaches~1 as $T_0 \to \infty$, since $S^{(t)}$ has mean $\mu_0 + \delta_{\mathrm{env}}$ while the control limits remain centered at $\mu_0$. This constitutes a false alarm (the change is legitimate), so static control fails adaptivity. \qed

\section{Moltbook Scale Simulations}
\label{app:scale}

We repeat the detectability analysis of Section~\ref{sec:experiments} at the scale of agent-only networks such as Moltbook~\citep{schlichtm2025moltbook}. We vary the number of agents $N \in \{10^2, 10^3, 10^4, 10^5\}$, the mean degree of the contact graph, the number of sleeper agents $N_{\mathrm{sleeper}} \in \{1, 10, 100\}$, and the size of the query pool $M \in \{50, 500\}$, which we refer to as information diversity.

\paragraph{Setup.}
We use a vectorized re-implementation of the proxy in Algorithm~\ref{alg:nollm-answer} ($p_{\text{adv-prop}} = 0.8$), checked against the original simulator. Agents interact over an Erd\H{o}s--R\'enyi contact graph with mean degree in $\{9, 99, 999, 9{,}999\}$, or over the complete graph (mean degree $N-1$). The query pool keeps the composition of the main experiments: $20\%$ environmental queries and $80\%$ static queries, with each static query seeded to exactly one agent. Agents use $k = 3$ retrieval, $K = 5$ questions per interaction and $\rho = 0.05$. Environmental answers change with probability $p_{\text{env}} = 0.04$ per iteration, and each run lasts $2{,}000$ iterations. Every sleeper agent follows the slow sigmoid schedule with $t^* = 400$, $s = 1{,}600$, $p_{\max} = 0.75$ and $\kappa = 1$.

For $N > 100$, TDKPS is computed on a fixed panel of $100$ agents that includes a sleeper agent; for $N = 100$ it uses all agents. The panel is queried every $10$ iterations on all environmental queries and an equal number of sampled static queries. The iso-mirror is obtained by one-dimensional Isomap, and monitored with an adaptive control chart with $w = 200$ and $k = 3$, evaluated at every snapshot from $t^*$ on ($160$ evaluations per run).

\paragraph{Metrics.}
Every configuration is run $30$ times with sleeper agents and $30$ times without, on matched seeds. \emph{Contamination} (``infected fraction of victims'' in Figure~\ref{fig:moltscale}) is the fraction of non-adversarial agents whose knowledge base holds the canary at the end of the run; it is the empirical counterpart of Corollary~\ref{cor:contamination}. The \emph{alarm ratio} is the fraction of chart evaluations out of band in the runs with sleepers, divided by the same fraction in the paired runs without them, pooled over the $30$ replications ($4{,}800$ evaluations per arm) with add-one smoothing; it is the empirical counterpart of condition~2 of Definition~\ref{def:undetected}.

\paragraph{Results.}
Figure~\ref{fig:moltscale} shows $100$ sleeper agents at $M = 500$. Contamination falls as the system grows and rises with connectivity. Five takeovers (contamination above one half) leave the monitor no more active than in a clean network (alarm ratio at most $1.0$ after rounding; outlined). In particular, $100$ sleeper agents in a $100{,}000$-agent system with mean degree 999, one tenth of one percent of the population, contaminate $57\%$ of the non-adversarial agents with an alarm ratio of $0.65$. Every other takeover has an alarm ratio above $1$, but these ratios (the $10^3$ column, and the $10^4$ column at mean degrees 99 and 999) rest on very few alarms, between $1$ and $4$ with sleepers against $0$ or $1$ without, and none is significant under a one-sided binomial test at level $0.05$. The legend's ``detectable'' label for these cells should be read in that light.

The figure is one of three configurations we ran. At the information diversity of the main experiments ($M = 50$), a single sleeper agent contaminates every non-adversarial agent at every network size and mean degree tested. There, for $N \ge 10^3$ the attack does raise the alarm rate measurably (ratios $1.4$ to $1.8$, each significant), so complete takeover coexists with a modest increase in alarms. At $M = 500$ with only $10$ sleepers (Table~\ref{tab:scale-10}), contamination is diluted as the system grows: at mean degree 9 it falls from $1.00$ at $N = 10^2$ to $0.01$ at $N = 10^5$, while down any column it recovers with connectivity. Raising the cohort to $100$ (Figure~\ref{fig:moltscale}) restores it. With $100$ sleepers at $N = 10\,N'$, contamination matches that of $10$ sleepers at $N'$ to within $0.01$ in every cell where both exist, so contamination tracks the adversarial fraction and the topology rather than the system size.

\begin{table}[h]
\centering
\small
\begin{tabular}{r cccc}
\toprule
 & \multicolumn{4}{c}{$N$ (agents)} \\
\cmidrule{2-5}
mean degree & $10^2$ & $10^3$ & $10^4$ & $10^5$ \\
\midrule
9      & 1.00 & 0.62 & 0.10 & 0.01 \\
99     & 1.00 & 0.95 & 0.26 & 0.03 \\
999    &  --  & 1.00 & 0.58 & 0.08 \\
9,999  &  --  &  --  & 0.70 & 0.11 \\
99,999 &  --  &  --  &  --  & 0.11 \\
\bottomrule
\end{tabular}
\caption{Contamination with $10$ sleeper agents at $M = 500$, otherwise as in Figure~\ref{fig:moltscale}; 30 replications per cell. Dashes mark mean degrees above $N-1$; the top entry of each column is the complete graph.}
\label{tab:scale-10}
\end{table}

\begin{figure}
    \centering
    \includegraphics[width=\linewidth]{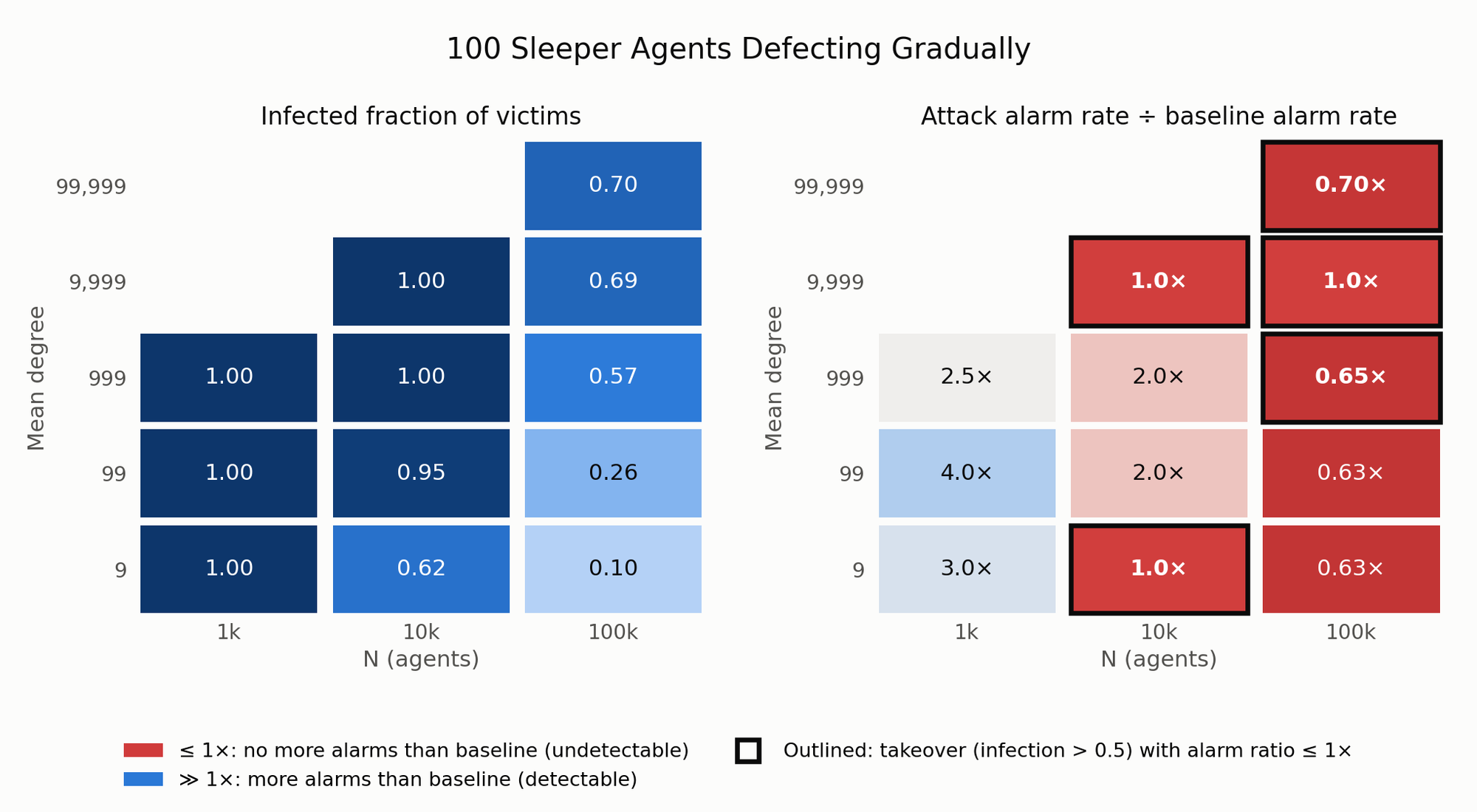}
        \caption{Gradual defection by $100$ sleeper agents across network size $N$ (columns) and mean degree (rows; the top cell in each column is the complete graph), with $M = 500$ queries, $t^* = 400$, $s = 1{,}600$ and $p_{\max} = 0.75$. TDKPS is computed on a fixed panel of $100$ agents that includes a sleeper agent. \textbf{Left}: Fraction of non-adversarial agents holding the canary at $T = 2{,}000$. \textbf{Right}: Alarm rate of the adaptive control chart ($w = 200$, $k = 3$) after $t^*$, relative to matched runs without sleeper agents; 30 replications per cell. Red: ratio at most 1. Outlined: takeovers (contamination above $0.5$) with ratio at most 1. The remaining takeovers have ratios above 1 that rest on at most $4$ alarms and are not significant (one-sided binomial test, level $0.05$). Gradual defection can take over large systems without raising the alarm rate.}
    \label{fig:moltscale}
\end{figure}
\clearpage

\end{document}